\documentclass[aps,twocolumn,showpacs,superscriptaddress]{revtex4}
\usepackage{amsthm}
\usepackage{amsmath}
\usepackage{latexsym}
\usepackage{amsfonts}
\usepackage{amssymb}
\usepackage{color}
\usepackage{bbm,dsfont}
\usepackage{graphicx}
\usepackage{hyperref}
\usepackage{enumerate}

\newtheorem{proposition}{Proposition}
\newtheorem{theorem}{Theorem}

\theoremstyle{definition}

\newcommand{\Z}{\mathbb Z} 
\newcommand{\half}{\tfrac{1}{2}} 
\newcommand{\mo}[1]{\left| #1 \right|} 

\newcommand{\abs}{\mo} 

\newcommand{\e}{{\rm e}} 

\newcommand{\ip}[2]{\left\langle\,#1\,|\,#2\,\right\rangle} 
\newcommand{\ket}[1]{|#1\rangle} 
\newcommand{\bra}[1]{\langle#1|} 
\newcommand{\kb}[2]{|#1\rangle\langle#2|} 
\newcommand{\no}[1]{\left\|#1\right\|} 
\newcommand{\tr}[1]{{\rm tr}\left[#1\right]} 

\newcommand{\id}{\mathbbm{1}} 

\renewcommand{\rho}{\varrho}

\newcommand{\va}{\mathbf{a}} 
\newcommand{\vb}{\mathbf{b}} 
\newcommand{\vc}{\mathbf{c}} 

\newcommand{\vk}{\mathbf{k}} 
\newcommand{\vsigma}{\boldsymbol{\sigma}} 

\newcommand{\vphi}{\mathbf{\phi}} 
\newcommand{\di}{\mathrm{d}}

\newcommand{\dotminus}{\stackrel{\textstyle\cdot}{\relbar}}

\begin{document}\setlength{\arraycolsep}{2pt}

\title[]{Determining quantum coherence with minimal resources}

\begin{abstract}
We characterize minimal measurement setups for validating the quantum coherence of an unknown quantum state. We show that for a $d$-level system, the optimal strategy consists of measuring $d$ orthonormal bases such that each measured basis is mutually unbiased with respect to the reference basis, and together with the reference basis they form an informationally complete set of measurements. We show that, in general, any strategy capable of validating quantum coherence allows one to evaluate also the exact value of coherence. We then give an explicit construction of the optimal measurements for arbitrary dimensions. Finally, we show that the same measurement setup is also optimal  for the modified task of verifying if the coherence is above or below a given threshold value. 
\end{abstract}


\author{Claudio Carmeli}
\email{claudio.carmeli@gmail.com}
\affiliation{DIME, Universit\`a di Genova, Via Magliotto 2, I-17100 Savona, Italy}

\author{Teiko Heinosaari}
\email{teiko.heinosaari@utu.fi}
\affiliation{Turku Centre for Quantum Physics, Department of Physics and Astronomy, University of Turku, FI-20014 Turku, Finland}

\author{Sabrina Maniscalco}
\email{smanis@utu.fi}
\affiliation{Turku Centre for Quantum Physics, Department of Physics and Astronomy, University of Turku, FI-20014 Turku, Finland}

\author{Jussi Schultz}
\email{jussi.schultz@gmail.com}
\affiliation{Turku Centre for Quantum Physics, Department of Physics and Astronomy, University of Turku, FI-20014 Turku, Finland}

\author{Alessandro Toigo}
\email{alessandro.toigo@polimi.it}
\affiliation{Dipartimento di Matematica, Politecnico di Milano, Piazza Leonardo da Vinci 32, I-20133 Milano, Italy}
\affiliation{I.N.F.N., Sezione di Milano, Via Celoria 16, I-20133 Milano, Italy}

\maketitle

\textit{Introduction.---}
Quantum coherence, or the ability to form superpositions of quantum states, is certainly one of the fundamental distinctions between the quantum and classical worlds. Quantum coherence is not merely a foundational curiosity, but it is the key element behind numerous quantum technological applications \cite{StAdPl16}, including quantum algorithms \cite{Hillery16} and quantum state merging \cite{StChRaBeWiLe16}. 
This has lead to the identification of quantum coherence as a true physical resource  \cite{WiYa16, StRaBeLe17,StAdPl16}. 

Due to the importance of quantum coherence, there have been various approaches and proposals for detecting or estimating the coherence of an unknown quantum state \cite{Girolami14,Wangetal17}.
In this paper, we address the problem of detecting coherence from a very fundamental point of view. 
 We determine the minimal number, as well as characterize the optimal set, of measurements for the following tasks: 
\begin{itemize}
\item[(a)] Certification of the presence of quantum coherence in an unknown quantum state.
\item[(b)] Determination of the exact value of quantum coherence in an unknown quantum state. 
\end{itemize}
While these tasks are clearly defined and provide a foundational basis for more applicative studies, one may wish to have a more robust and physically motivated goal for comparison. 
Hence we also consider the following task:
\begin{itemize}
\item[(c)] Verification that an unknown quantum state has more coherence than some chosen threshold value.  
\end{itemize}

Mathematically speaking, quantum coherence is always defined with respect to some \emph{reference basis}, that is, a fixed orthonormal basis $\{\varphi_j\}_{j=0}^{d-1}$ of a $d$-level quantum system. A state $\varrho$ represented by a density matrix is \emph{incoherent} if it is diagonal in the reference basis, i.e., $\varrho = \sum_j p_j \kb{\varphi_j}{\varphi_j}$ for some probability distribution $(p_j)$; otherwise it is \emph{coherent}. Quantum coherence can be quantified in various ways \cite{ChGo16,NaBrCiPiJoAd16,Rastegin16}. 
It turns out that task (b) does not depend on a particular choice of quantification; as a measure of coherence we can take any function which depends only on the off-diagonal elements of the state, and vanishes for all incoherent states and only for them. 
For task (c), we need to be more specific and fix a suitable measure. 
A natural choice is the $\ell_1$-norm of coherence, given as $C_1(\varrho)=\min_{\sigma\in\mathcal{I}} \no{\varrho - \sigma}_1=\sum_{j \neq k} \mo{\varrho_{jk}}$, where $\mathcal{I}$ is the set of all incoherent states.
The $\ell_1$-norm of coherence is a reasonable quantification of coherence  since the first formula gives a physically motivated definition while the second formula provides an easy way to calculate it.
More importantly, it satisfies the properties that a proper measure of coherence should satisfy, in particular, monotonicity under incoherent completely positive and trace preserving maps \cite{BaCrPl14}. 
It has been recently shown that the distillable coherence of a state is bounded by the $\ell_1$-norm of coherence, and this provides a solid operational interpretation for its meaning \cite{RaPaWiLe16}.

In this work we concentrate on quantum measurements related to orthonormal bases, so that the measurement of a basis $\{\psi_k\}_{k=0}^{d-1}$ on an initial state $\varrho$ gives an outcome $k$ with probability $\ip{\psi_k}{\varrho \psi_k}$. A collection of finitely many different measurements is referred to as a \emph{measurement setup}. To put our questions into a proper context, we recall that in order to perform full state tomography, a measurement setup consisting of $d+1$ bases is required; such a setup is called \emph{informationally complete} \cite{Prugovecki77,BuLa89}. 
Informationally complete setups exist in all dimensions, and even a randomly chosen setup works almost always as the only criterion is that the respective projections span the linear space of operators \cite{Busch91,SiSt92}. 
The main goal of the present investigation is therefore to determine how many measurements are needed in tasks (a), (b) and (c), and to characterize the respective minimal measurement setups.

To further motivate the stated tasks, we recall that it is known that the certification of purity can be done with just five measurements \cite{ChDaJietal13}, whereas the determination of the exact value of purity requires an informationally complete setup \cite{CaHeScTo17}. 
Another recent result regarding quantum correlations shows that any measurement setup which is not informationally complete will fail in certifying entanglement as well as nonclassicality of an unknown state \cite{CaHeKaScTo16}. 
It is therefore interesting to compare how coherence compares with these other properties of quantum states. 

Our main result shows that any measurement setup which can complete either task (a) or (c) is also capable of completing the more demanding task (b).
We will prove that the minimal measurement setup for all tasks (a), (b) and (c) consists of $d$ bases.
The essential property that any measurement setup of this kind must satisfy is that each basis is mutually unbiased with regards to the reference basis. 
We recall that the mutual unbiasedness of two bases $\{\psi_k\}_{k=0}^{d-1}$ and $\{\phi_\ell\}_{\ell=0}^{d-1}$ means that any pair of vectors $\psi_k$ and $\phi_\ell$ satisfy $\mo{\ip{\psi_k}{\phi_\ell}}=1/\sqrt{d}$.

It is known that if the dimension $d$ is a prime power, then one can construct a complete set of mutually unbiased bases (MUB), i.e., $d+1$ bases that are all pairwisely mutually unbiased \cite{WoFi89}. 
In prime power dimension one could therefore pick a complete set of MUB, transform all of them with a suitable unitary operator so that one of the bases becomes the reference basis, and then simply drop out the reference basis. However, as the existence of a complete set of MUB in other dimensions still remains an open question \cite{DuEnBeZy10}, this does not yet give a satisfactory general solution. 
Moreover, it should be emphasized that the criterion for a minimal measurement setup does not require all pairs of bases to be mutually unbiased with respect to each other, so a positive answer to the MUB existence question is not a prerequisite to our result. 
In fact, we provide an explicit construction of the minimal measurement setup in arbitrary dimensions.

\vspace{0.5cm}

\textit{Coherence of a qubit state.---}
As a warm up, we consider tasks (a) and (b) for a qubit system. 
This simple example already indicates the connection of these tasks to the geometry of successful measurement setups.
For a qubit state $\varrho$ we have $C_1(\varrho)=2\mo{\varrho_{01}}$, where $\varrho_{01}$ is the off-diagonal matrix element in the chosen reference basis.
Assuming that this basis is the eigenbasis of $\sigma_z$, we see that from the numbers $\tr{\varrho \sigma_x}$ and $\tr{\varrho \sigma_y}$ we can calculate the off-diagonal elements and hence the value of $C_1(\varrho)$. Therefore, task (b) can be accomplished by measuring just two bases, namely, the eigenbases of $\sigma_x$ and $\sigma_y$. It is clear that a single basis is not enough even for task (a), since the measurement of, say, $\sigma_x$ gives the same outcome probabilities for the coherent state $\half (\id + \sigma_y)$ and the incoherent state $\half \id$.

Let us now look at the possible choices for the two measurements. To this end, let $\va$ and $\vb$ be two nonparallel unit vectors in the Bloch ball, and assume that at least one of them is \emph{not} orthogonal to the unit vector $\vk$ ($z$-direction). 
We can set $\vc = \va \times \vb / \no{\va \times \vb}$ to obtain another unit vector which is not parallel to $\vk$.
The state $\varrho_\vc = \half( \id + \vc \cdot \vsigma)$ is then coherent, but since $\vc$ is orthogonal to $\va$ and $\vb$ we have
\begin{equation}
\tr{\varrho_\vc \sigma_\va} = \tr{\half \id \sigma_\va} , \, \tr{\varrho_\vc \sigma_\vb} = \tr{\half \id \sigma_\vb} \, , 
 \end{equation}
 which means that the eigenbases of $\sigma_\va$ and $\sigma_\vb$ cannot distinguish the states  $\varrho_\vc$ and $\half \id$.
 This kind of ambivalence does not occur if both $\va$ and $\vb$ are orthogonal to $\vk$. 
The requirement for orthogonality in the Bloch picture means that the respective orthonormal bases must be mutually unbiased with respect to the reference basis. 

\vspace{0.5cm}
\textit{Perturbation operators.---}
In order to deal with task (a) in the case of a general $d$-level system, we resort to a geometric framework, similar to that used in \cite{CaHeKaScTo16,CaHeScTo17}. The key concept is that of a \emph{perturbation operator}, by which we mean any traceless selfadjoint operator. The nonzero perturbation operators $\Delta$ are precisely those operators that can be written, up to a scaling factor, as differences of distinct quantum states $\delta\Delta = \varrho' - \varrho $.  Therefore, they provide  a convenient way to analyze what kinds of pairs of states a given measurement can distinguish. Equivalently, by solving for $\varrho' = \varrho + \delta \Delta$, we see that we are in fact studying if the state $\varrho$ can be distinguished from its slightly perturbed version $\varrho + \delta \Delta$.

We say that a  perturbation operator $\Delta$ is \emph{detected} by a measurement setup if for any state $\varrho$, the two states $\varrho$ and $\varrho + \delta\Delta$ give different measurement data whenever the scaling parameter $\delta$ is such that the latter operator is a valid state; otherwise $\Delta$ is \emph{undetected}.
We see that for a measurement setup consisting of $m$ bases $\{\psi^{(1)}_k\}_{k=0}^{d-1}$, $\ldots$, $\{\psi^{(m)}_k\}_{k=0}^{d-1}$, the undetected perturbation operators $\Delta$ are exactly those that satisfy
\begin{equation}\label{eq:delta0}
\ip{\psi_k^{(\ell)}}{\Delta \psi_k^{(\ell)}}=0
\end{equation}
for all $k=0,\ldots,d-1$ and $\ell=1,\ldots,m$.

For certifying coherence (task (a)), a suitable measurement setup must be able to detect differences between coherent and incoherent states. Quite naturally, this means that all perturbation operators with at least some nonzero off-diagonal elements must be detected. To see this, note that for any perturbation $\Delta$, there exists a $\delta >0$ such that $\eta:=\tfrac{1}{d} \id + \delta \Delta$ is a valid state (i.e., positive). Clearly $\eta$ is coherent if and only if $\Delta$ has nonzero off-diagonal elements, and since $\tfrac{1}{d} \id$ is incoherent, such $\Delta$'s must be detected. 

Since diagonal perturbation operators can never be written as differences of coherent and incoherent states, their detection is irrelevant to us. This allows us to summarize the above discussion in the following statement.

\begin{theorem}\label{thm:diagonal}
A measurement setup consisting of $m$ bases $\{\psi^{(1)}_k\}_{k=0}^{d-1}$, $\ldots$, $\{\psi^{(m)}_k\}_{k=0}^{d-1}$ completes task (a) if and only if all the undetected perturbation operators are diagonal in the reference basis.
\end{theorem}

This result gives a lower bound for the minimal number $m$.  First of all, we observe that the undetected perturbations form a subspace of traceless selfadjoint operators. The dimension of this subspace cannot exceed $d-1$,  since this is the  dimension of the subspace of all diagonal perturbation operators. Equation \eqref{eq:delta0}  can be equivalently written as
\begin{equation}\label{eq:delta0hs}
\tr{\kb{\psi^{(\ell)}_k}{\psi^{(\ell)}_k} \Delta} = 0 \,  ,
\end{equation}
which means that the measured projections $\kb{\psi^{(\ell)}_k}{\psi^{(\ell)}_k}$ and the undetected perturbations $\Delta$ are orthogonal in the Hilbert-Schmidt inner product. Since the dimension of the real vector space of all selfadjoint operators is $d^2$, the required measurement setup must span a subspace of at least dimension $d^2-d+1$.
Since $m$ bases give at most $d+(m-1)(d-1)$ linearly independent operators, 
this means that task (a) cannot be solved with less than $d$ orthonormal bases.

\vspace{0.5cm}
\textit{Equivalence of tasks (a) and (b).---}
We will next explain that, quite surprisingly, any measurement setup that completes task (a) also completes the more demanding task (b).
Assume that we have a measurement setup $\{\psi^{(1)}_k\}_{k=0}^{d-1}$, $\ldots$, $\{\psi^{(m)}_k\}_{k=0}^{d-1}$ that completes task (a). As we have earlier concluded, any undetected perturbation operator $\Delta$ must be diagonal, and hence in the linear span of the operators 
$$
D_{j} =  \kb{\varphi_j}{\varphi_j} - \kb{\varphi_{j-1}}{\varphi_{j-1}} \, ,  
$$
where $1\leq j \leq d-1$. For all indices $0\leq j<k \leq d-1$, we further denote
\begin{align*}
A^+_{jk} &= \half(\kb{\varphi_j}{\varphi_k} + \kb{\varphi_k}{\varphi_j})\, ,  \\
 A^-_{jk} &= \tfrac{i}{2}(\kb{\varphi_j}{\varphi_k} - \kb{\varphi_k}{\varphi_j})\,.
\end{align*}
Since the operators $A^+_{jk}$ and $A^-_{jk}$ are orthogonal to all $D_i$'s, they are also orthogonal to all of the undetected perturbations $\Delta$. Hence  the operators $A^+_{jk}$ and $A^-_{jk}$ must be in the linear span of the projections $\kb{\psi_k^{(\ell)}}{\psi_k^{(\ell)}}$, and the expectations $\tr{\varrho A^+_{jk}}$, $\tr{\varrho A^-_{jk}}$ can be evaluated from the probabilities $\ip{\psi_k^{(\ell)}}{\varrho \psi_k^{(\ell)}}$.

We now observe that for any $j<k$, we can write
\begin{align*}
\varrho_{jk} = \tr{\varrho A^+_{jk}} + i\,\tr{\varrho A^-_{jk}}\, .
\end{align*} 
By our previous observation, this means that we can calculate the off-diagonal elements of $\varrho$ from the measurement data and hence calculate $C_1(\varrho)$ or any other measure of coherence. 
We summarize this discussion in the following theorem.

\begin{theorem}\label{thm:equivalence}
Any measurement setup that completes task (a) also completes task (b).
\end{theorem}

\vspace{0.5cm}
\textit{Minimality implies mutual unbiasedness.---}
We have earlier concluded that task (a) cannot be completed with less than $d$ orthonormal bases. 
Although we have not yet shown that a suitable set of $d$ bases exists, we will next see what implications this would have. 
To this end, suppose that there is a set of $d$ bases $\{\psi^{(1)}_k\}_{k=0}^{d-1}$, $\ldots$, $\{\psi^{(d)}_k\}_{k=0}^{d-1}$ which completes task (a). 
By counting the dimension of the corresponding subspace of undetected perturbation operators, 
we see that  \eqref{eq:delta0} must hold for all indices $k,\ell$ \emph{whenever} $\Delta$ is a diagonal perturbation operator. 
In other words, there is no room to detect any perturbations that are not relevant for the task at hand.

Since the operators $\kb{\varphi_i}{\varphi_i} - \kb{\varphi_{j}}{\varphi_{j}}$ with $i\neq j$ span the subspace of diagonal perturbation operators, we find that \eqref{eq:delta0} is equivalent to
\begin{equation}
\mo{\ip{\psi_k^{(\ell)}}{\varphi_i}} = \mo{\ip{\psi_k^{(\ell)}}{\varphi_j}} \,.
\end{equation}
But this is nothing else than the mutual unbiasedness condition for the reference basis and each measurement basis.
Therefore, if a collection of $d$ bases completes task (a), then each basis must be mutually unbiased with respect to the reference basis.

There is an additional property that we can infer from the assumed minimal set of $d$ bases. 
The reference basis clearly detects all diagonal perturbations.
Hence, if we measure the reference basis together with the assumed set of $d$ bases, we can detect all perturbations.
We summarize this discussion in the following theorem.

\begin{theorem}\label{thm:conditions}
A minimal measurement setup consisting of $d$ bases $\{\psi^{(1)}_k\}_{k=0}^{d-1}$, $\ldots$, $\{\psi^{(d)}_{k}\}_{k=0}^{d-1}$ completes task (a) (equivalently (b)) if and only if 
\begin{itemize}
\item[(i)] each basis is mutually unbiased with respect to the reference basis; and 
\item[(ii)] together with the reference basis they form an informationally complete set.
\end{itemize}
\end{theorem}

The remaining task is to show that $d$ orthonormal bases with these two required properties can be constructed in all dimensions.

\vspace{0.5cm}
\textit{Construction of a minimal setup.---}
In the prime power dimensions, the construction of a minimal measurement setup is rather straightforward. Namely, we can pick a complete set of $d+1$ MUB, apply a unitary transformation which transforms one of the bases into the reference basis, and then drop out the reference basis. It is an immediate consequence of Theorem \ref{thm:conditions} that the remaining set of $d$ MUB completes task (a) and, as we have previously seen, also (b). For other dimensions we need to seek a different construction. However, one can expect this to be possible since (i) is a seemingly weaker condition than having a complete set of MUB.

First, let us recall a simple way to write a mutually unbiased basis with respect to a given one \cite{BaBoRoVa02}.
We denote $\omega = e^{2\pi i/d}$.
For each $j=0,\ldots,d-1$, we fix a complex number $\beta_j$ with $\mo{\beta_j}=1$.
Then, for each $k=0,\ldots,d-1$, we define a unit vector
\begin{equation}
\psi_k = \frac{1}{\sqrt{d}} \sum_{j=0}^{d-1} \beta_j \omega^{jk}\varphi_j \, .
\end{equation}
It is straightforward to check that $\{\psi_k\}_{k=0}^{d-1}$ is an orthonormal basis and that it is unbiased with respect to the reference basis $\{\varphi_j\}_{j=0}^{d-1}$.
Since the choice of the numbers $\beta_j$ is arbitrary, we can construct any number of different orthonormal bases that are mutually unbiased with respect to the reference basis. 
However, this is not enough as we also need to satisfy the condition (ii) in Theorem \ref{thm:conditions}.
We thus need to choose the numbers $\beta_j$ in a specific way.

For the construction of the bases, we need some additional notation.
We denote by $\Z_d$ the ring of integers $\{0,1,\ldots,d-1\}$, where the addition and multiplication is understood modulo $d$.
For clarity, we denote by $\dotplus$ the addition modulo $d$.
Unless $d$ is prime, $\Z_d$ has zero divisors, i.e., nonzero elements that give 0 when multiplied with some other nonzero element (e.g. in $\Z_6$ we have $2 \cdot 3 = 0$ modulo $6$).  
This is the underlying fact why there are relatively easy constructions of complete set of MUB in prime dimensions but similar constructions do not work if the dimension is a composite number. 
For this reason, we define a map $s:\Z_d \to \Z$ that maps $x\in\Z_d$ into $x\in\Z$. 
The purpose of the map $s$ is that we can use the modular arithmetic of $\Z_d$ when needed, but then perform calculations in $\Z$ when that is more convenient as $\Z$ has no zero divisors.

We fix an irrational number $\alpha$.
For each $\ell=1,\ldots,d$, we then define an orthonormal basis $\psi_k^{(\ell)}$ by
\begin{equation}\label{eq:bases}
\psi_k^{(\ell)} = \frac{1}{\sqrt{d}}\sum_{j\in\Z_d} \e^{\alpha\pi i (\ell-1) s(j)^2}  \omega^{jk}\varphi_j \, .
\end{equation}
By Theorem \ref{thm:diagonal}, we need to show that any undetected perturbation operator for this measurement setup is diagonal.
To do this, we write an arbitrary perturbation operator as $\Delta = \sum_{i,j} \nu_{i,j} \kb{\varphi_i}{\varphi_j}$.
By then taking the discrete Fourier transform of \eqref{eq:delta0} with respect to index  $k$, we obtain
\begin{equation}\label{eq:fourier}
\sum_{j\in\Z_d}  \e^{\alpha\pi i (\ell-1) [s(j \dotplus x)^2-s(j)^2]} \nu_{j,j \dotplus x}= 0 \, , 
\end{equation}
which is required to hold for all $x\in\Z_d$ and $\ell =1,\ldots,d$.
Since we need to show that $\nu_{j,j \dotplus x}=0$ for any $x\neq 0$, we fix $x\neq 0$ and consider \eqref{eq:fourier} for $\ell =1,\ldots,d$ as a system of linear equations. 
We thus need to show that the corresponding matrix is invertible, as then $\nu_{j,j \dotplus x}\equiv 0$ is the only solution. 
We denote
\begin{align*}
u_x(j) =\e^{-\alpha\pi i [s(j \dotplus x)^2-s(j)^2]} \, , 
\end{align*}
and observe that the matrix of \eqref{eq:fourier} is a Vandermonde matrix.
Hence, its determinant can be written as
\begin{equation}\label{eq:det}
\prod_{0\leq i < j \leq d-1} (u_x(j)-u_x(i)) \, .
\end{equation}
The last effort is to verify that all the factors in this product are nonzero.
Since $\alpha$ is irrational, we see that $u_x(j)=u_x(i)$ if and only if 
\begin{equation}\label{eq:s}
s(j \dotplus x)^2-s(j)^2 = s(i \dotplus x)^2-s(i)^2 \, .
\end{equation}
It can be shown that the only solution to \eqref{eq:s} is $i=j$.
The detailed mathematical calculations are given in the Supplemental Material.
As the expression \eqref{eq:det} contains only terms with $i\neq j$, this completes the proof.
Therefore, we conclude that the measurement setup consisting of the $d$ orthonormal bases written in \eqref{eq:bases} completes tasks (a) and (b).
In addition, one can write a reconstruction formula for the off-diagonal elements using only the measurement outcome probabilities of these bases. 
An explicit formula is given in the Supplemental Material.

\vspace{0.5cm}
\textit{Completion of task (c).---}
In order to make task (c) more precise, we need to use a specific measure of coherence.
We will use the $\ell_1$-norm of coherence, defined as \(C_1(\varrho)=\sum_{j \neq k} \mo{\varrho_{jk}}\). 
This measure clearly vanishes exactly for incoherent states, and the maximal value of $C_1$ is $d-1$ \cite{BaCrPl14}.
As we did earlier in the case of tasks (a) and (b), we aim to determine the minimal requirement for a measurement setup to be capable of deciding,  for any state \(\rho\), whether \(C_1(\rho)\) is greater than a fixed threshold value \(r(d-1)\)  with \(0<r<1\), or not. 

In the following we show that a measurement setup completes the stated task if and only if it detects the nondiagonal perturbation operators. 
Indeed, using the earlier framework of perturbation operators, if \(C_1(\rho)\leq r(d-1)\) and \(C_1(\varrho')>r(d-1)\), it is clear that the difference $\Delta = \varrho'-\varrho$ is a nondiagonal perturbation operator. Therefore, it is enough to prove that for any nondiagonal perturbation operator $\Delta$, there exist states $\varrho$ and $\varrho' = \varrho + \delta\Delta$ with \(C_1(\rho)\leq r(d-1)\) and \(C_1(\varrho')>r(d-1)\). 
This result will then lead to the following analogue of Theorem~\ref{thm:diagonal}, thus establishing the equivalence of task (c) with (a) and (b).

\begin{theorem}\label{thm:diagonal2}
A measurement setup consisting of $m$ bases $\{\psi^{(1)}_k\}_{k=0}^{d-1}$, $\ldots$, $\{\psi^{(m)}_k\}_{k=0}^{d-1}$ completes task (c) if and only if all the undetected perturbation operators are diagonal in the reference basis.
\end{theorem}

To prove this statement, we fix a family of maximally coherent states
\(
\ket{\psi}_{{\vphi}}= \frac{1}{\sqrt{d}}\sum_{k=1}^d e^{i \phi_k} \ket{k}
\), parametrized by an array of phases \(\vphi=(\phi_1,\ldots,\phi_d)\). 
The convex mixture
\begin{equation}
\rho_{\vphi}^r = (1-r) \frac{\id}{d} + r \ket{\psi}_{\vphi} \bra{\psi}_{\vphi} \, , \quad 0<r<1 \, , 
\end{equation}
is then a full-rank state with \(C_1(\rho_{\vphi}^r )=r(d-1)\). 
It follows that we can choose \(\delta\) small enough so that \(\rho_{\vphi}^r +(\delta r/d)\Delta\)  is a (full-rank) state. 
Namely, we choose any $\delta$ such that \(\abs{\delta} < (1-r)/(r\no{\Delta})\). 
The remaining step is to prove that, in addition, one can always choose \(\delta\) such that 
\begin{equation}
C_1(\rho_{\vphi}^r +(\delta r/d)\Delta)>r(d-1) \, .
\end{equation}
We prove this fact by showing that the function \(f_{\vphi}(\delta)=C_1(\rho_{\vphi}^r +(\delta r/d)\Delta)\) has a nonzero derivative at the point \(\delta=0\) for some set of parameters \(\vphi_0\). 
Possibly replacing $\delta$ with $-\delta$, this then yields $f_{\vphi_0}(\delta) > f_{\vphi_0}(0)$ for some $\delta$, which is what we want to show. 

Let \(\Delta=A+i B\) be the decomposition of \(\Delta\) into its real symmetric and antisymmetric parts. 
A straightforward computation gives
\begin{align*}
\label{eq:esprcoh2}
f_\vphi(\delta) =C_1\left(\rho_{\vphi}^r +\delta \frac{r}{d} \Delta\right) & =  \frac{r}{d}\sum_{p\neq q} \left[(\cos{(\phi_p-\phi_q)}+\delta A_{pq})^2\right.
\\ 
& \left.  +
(\sin{(\phi_p-\phi_q)}+\delta B_{pq})^2
\right]^{\frac12} \, .
\end{align*}
Therefore,
\begin{align*}
g(\vphi)=\left. f_\vphi^\prime(\delta)\right|_{\delta=0} &  = \frac{2r}{d}\sum_{p>q} \left[ A_{pq}\cos{(\phi_p-\phi_q)} \right.\notag\\
&\left.+ B_{pq} \sin(\phi_p-\phi_q)\right] \, .
\end{align*}
We can now use the freedom in the choice of \(\vphi\) to obtain the result. Indeed, if \(r>s\) we have
\begin{align*}
\int_{[0,2\pi)^2} g(\vphi)\cos(\phi_r-\phi_s) \di \phi_r\di\phi_s & = \frac{(2\pi)^2 r}{d} A_{pq} \\
\int_{[0,2\pi)^2} g(\vphi)\sin(\phi_r-\phi_s) \di \phi_r\di\phi_s & = \frac{(2\pi)^2 r}{d} B_{pq}
\end{align*}
Hence, if \(g(\vphi)=0\) for every \(\vphi\in [0,2\pi)^d\), then \(A_{pq}=B_{pq}=0\) for every \(p\neq q\), and \(\Delta\) is diagonal, against our assumption.
 
\vspace{0.5cm}
\textit{Discussion.---}
The development of quantum technologies will bring us applications capable of outperforming any of their classical counterparts. The superiority of these applications rests on the ability to take advantage of properties of physical systems which are genuinely quantum. For this reason it is essential to be able to verify that a given source produces systems which have such a property. Here we have investigated optimal measurement strategies for verifying the presence of quantum coherence. We have shown that this simple verification task is actually as difficult as determining the exact value of quantum coherence. We have both characterized the optimal setups in terms of a mutual unbiasedness condition, as well as constructed explicit examples in arbitrary dimensions. 

One of the core assumptions behind our results is that there is no prior information available regarding the initial state of the system at hand - its quantum state is completely unknown. In many practical situations this may not be the case, and by exploiting the available prior information it may be possible to further optimize the setup. As a simple example, suppose that we know the system to be in a pure state. Then quantum coherence can be verified by simply measuring the reference basis, as the incoherent states are exactly the eigenstates of this observable. 
The geometric framework exploited in this work is flexible enough to be used also in questions with prior information.

\vspace{0.5cm}
\textit{Acknowledgement.---}
The authors thank Tom Bullock for his comments on an earlier version of this paper. 
T.H. and S.M. acknowledge financial support from the Horizon 2020 EU collaborative projects QuProCS (Grant Agreement No. 641277), the Academy of Finland (Project no. 287750), and the Magnus Ehrnrooth Foundation.

\onecolumngrid

\newpage

\section*{SUPPLEMENTAL MATERIAL}

\subsection{Proof that Equation \eqref{eq:s} for $x\neq 0$ implies $i=j$.}

The function $s$ is the main text was defined as $s:\Z_d \to \Z$, $s(x)=x$.
For two elements $x,y\in\Z_d$, we have
\begin{equation}\label{eq:x-y}
s(x \dotminus y) = \begin{cases}
s(x)-s(y) & \text{ if $s(x)\geq s(y)$} \, , \\
s(x)-s(y)+d & \text{ if $s(x)<s(y)$} \, ,
\end{cases}
\end{equation}
where $\dotminus$ denotes the subtraction modulo $d$.

The following result is the last step needed in the construction of $d$ bases.

\begin{proposition}
Let $x,i,j\in\Z_d$ and $x\neq 0$.
The only solution to
\begin{equation}\label{eq:s-2}
s(j \dotplus x)^2-s(j)^2 = s(i \dotplus x)^2-s(i)^2 \, .
\end{equation}
is $i=j$.
\end{proposition}

\begin{proof}
One can prove the claim by considering the different possible cases separately, i.e., $s(j \dotplus x) \geq s(j)$ and $s(j \dotplus x) < s(j)$. 
We go through the first case, the latter being similar. 

Assume $s(j \dotplus x) \geq s(j)$.
Then by \eqref{eq:s-2} we also have $s(i \dotplus x) \geq s(i)$.
Using \eqref{eq:x-y} we can thus write $s(j \dotplus x)  - s(j) = s(x)$ and $s(i \dotplus x)  - s(i) = s(x)$, and \eqref{eq:s-2} becomes
\begin{equation}\label{eq:s-proof-1}
s(x) ( s(j \dotplus x) + s(j) ) = s(x) ( s(i \dotplus x) + s(i) ) \, .
\end{equation}
As $s(x) \neq 0$, we further get
\begin{equation}\label{eq:s-proof-2}
s(j \dotplus x) + s(j) = s(i \dotplus x) + s(i) \, .
\end{equation}
We then split the proof further into two cases. 
\begin{itemize}
\item[(a)] Assume $s(j) \geq s(i)$. 
From \eqref{eq:s-proof-2}, it follows that $s(i \dotplus x) \geq s(j \dotplus x)$.
Using \eqref{eq:x-y} and \eqref{eq:s-proof-2}, we obtain
\begin{equation}\label{eq:s-proof-3.1}
s(j \dotminus i) = s(i \dotminus j) \, ,
\end{equation}
and further
\begin{equation}\label{eq:s-proof-4.1}
2(i-j) = 0 \text{ modulo $d$}
\end{equation}
due to the injectivity of $s$.
If $d$ is odd, then \eqref{eq:s-proof-4.1} immediately implies $i=j$.
If $d$ is even, then either $i=j$ or $i \dotminus j = d/2$ modulo $d$.
However, in the latter option we have either $s(i)\geq d/2$, $s(j)<d/2$, or $s(i) < d/2$, $s(j) \geq d/2$. These are not consistent with \eqref{eq:s-proof-2} and the assumed conditions $s(j \dotplus x) \geq s(j)$ and $s(i \dotplus x) \geq s(i)$.
We thus conclude that $i=j$.
\item[(b)] Assume $s(j) < s(i)$. From \eqref{eq:s-proof-2} follows that $s(i \dotplus x) < s(j \dotplus x)$.
Using \eqref{eq:x-y} and \eqref{eq:s-proof-2}, we obtain \eqref{eq:s-proof-3.1} and \eqref{eq:s-proof-4.1} as in the previous case. Then, the rest of the proof is identical.
\end{itemize}

\end{proof}

\subsection{Reconstruction formula for off-diagonal elements.}

The $d$ basis vectors are
\begin{equation}
\psi_k^{(\ell)} = \frac{1}{\sqrt{d}}\sum_{j\in\Z_d} \e^{\alpha\pi i (\ell-1) s(j)^2}  \omega^{jk}\varphi_j \qquad k=0,\ldots,d-1
\end{equation}
for $\ell=1,\ldots,d$. The corresponding probability distributions in the state $\varrho$ are
\begin{equation}\label{eq:supplement-delta0}
p^{(\ell)}(k) = \ip{\psi_k^{(\ell)}}{\varrho \psi_k^{(\ell)}} \qquad k=0,\ldots,d-1 \,.
\end{equation}
We have
\begin{equation}
\sum_{k=0}^{d-1} \omega^{-kz} p^{(\ell)}(k) = \sum_{h=0}^{d-1} V_{\ell-1,h}\varrho_{h,h\dot{+}z}
\end{equation}
where $(V_{j,h})_{j,h = 0,\ldots,d-1}$ is the Vandermonde matrix
\begin{equation}
V_{j,h} = \left(\frac{\e^{\alpha\pi i s(h\dot{+}z)^2}}{\e^{\alpha\pi i s(h)^2}}\right)^j = x_h^j \qquad \text{where} \qquad x_h = \frac{\e^{\alpha\pi i s(h\dot{+}z)^2}}{\e^{\alpha\pi i s(h)^2}} \,.
\end{equation}
Its inverse is \cite{MS58}
\begin{equation}
(V^{-1})_{h,j} = \frac{x_h (-1)^{d-j}}{\prod_{i\neq h}(x_h-x_i)} \sigma^h_{d-j,d-1}
\end{equation}
where $\sigma^h_{i,d-1}$ is the sum of all products of $i$ of the numbers $x_1,\ldots,x_{h-1},x_{h+1},\ldots,x_{d-1}$ without permutations or repetitions ($\sigma^h_{0,d-1} \equiv 1$). 
Therefore,
\begin{eqnarray}
\varrho_{h,h\dot{+}z} &=& \sum_{j,k=0}^{d-1} (V^{-1})_{h,j}\, \omega^{-kz} p^{(j+1)}(k) \\
&=& \sum_{j,k=0}^{d-1} \frac{x_h (-1)^{d-j}\omega^{-kz}}{\prod_{i\neq h}(x_h-x_i)} \sigma^h_{d-j,d-1} p^{(j+1)}(k) \,.
\end{eqnarray}


\begin{thebibliography}{10}

\bibitem{StAdPl16}
M.B.~Plenio A.~Streltsov, G.~Adesso.
\newblock Quantum coherence as a resource.
\newblock arXiv:1609.02439 [quant-ph], 2016.

\bibitem{Hillery16}
M.~Hillery.
\newblock Coherence as a resource in decision problems: The {D}eutsch-{J}ozsa
  algorithm and a variation.
\newblock {\em Phys. Rev. A}, 93:012111, 2016.

\bibitem{StChRaBeWiLe16}
A.~Streltsov, E.~Chitambar, S.~Rana, M.N. Bera, A.~Winter, and M.~Lewenstein.
\newblock Entanglement and coherence in quantum state merging.
\newblock {\em Phys. Rev. Lett.}, 116:240405, 2016.

\bibitem{WiYa16}
A.~Winter and D.~Yang.
\newblock Operational resource theory of coherence.
\newblock {\em Phys. Rev. Lett.}, 116:120404, 2016.

\bibitem{StRaBeLe17}
A.~Streltsov, S.~Rana, M.N. Bera, and M.~Lewenstein.
\newblock Towards resource theory of coherence in distributed scenarios.
\newblock {\em Phys. Rev. X}, 7:011024, 2017.

\bibitem{Girolami14}
D.~Girolami.
\newblock Observable measure of quantum coherence in finite dimensional
  systems.
\newblock {\em Phys. Rev. Lett.}, 113:170401, 2014.

\bibitem{Wangetal17}
Yi-Tao Wang, Jian-Shun Tang, Zhi-Yuan Wei, Shang Yu, Zhi-Jin Ke, Xiao-Ye Xu,
  Chuan-Feng Li, and Guang-Can Guo.
\newblock Directly measuring the degree of quantum coherence using interference
  fringes.
\newblock {\em Phys. Rev. Lett.}, 118:020403, 2017.

\bibitem{ChGo16}
E.~Chitambar and G.~Gour.
\newblock Comparison of incoherent operations and measures of coherence.
\newblock {\em Phys. Rev. A}, 94:052336, 2016.

\bibitem{NaBrCiPiJoAd16}
C.~Napoli, T.R. Bromley, M.~Cianciaruso, M.~Piani, N.~Johnston, and G.~Adesso.
\newblock Robustness of coherence: An operational and observable measure of
  quantum coherence.
\newblock {\em Phys. Rev. Lett.}, 116:150502, 2016.

\bibitem{Rastegin16}
A.E. Rastegin.
\newblock Quantum-coherence quantifiers based on the {T}sallis relative
  $\ensuremath{\alpha}$ entropies.
\newblock {\em Phys. Rev. A}, 93:032136, 2016.

\bibitem{BaCrPl14}
T.~Baumgratz, M.~Cramer, and M.B. Plenio.
\newblock Quantifying coherence.
\newblock {\em Phys. Rev. Lett.}, 113:140401, 2014.

\bibitem{RaPaWiLe16}
S.~Rana, P.~Parashar, A.~Winter, and M.~Lewenstein.
\newblock Logarithmic coherence: Operational interpretation of
  $\mathbf{\ell_1}$-norm coherence.
\newblock arXiv:1612.09234 [quant-ph].

\bibitem{Prugovecki77}
E.~Prugove\v{c}ki.
\newblock Information-theoretical aspects of quantum measurements.
\newblock {\em Int. J. Theor. Phys.}, 16:321--331, 1977.

\bibitem{BuLa89}
P.~Busch and P.~Lahti.
\newblock The determination of the past and the future of a physical system in
  quantum mechanics.
\newblock {\em Found. Phys.}, 19:633--678, 1989.

\bibitem{Busch91}
P.~Busch.
\newblock Informationally complete sets of physical quantities.
\newblock {\em Int. J. Theor. Phys.}, 30:1217--1227, 1991.

\bibitem{SiSt92}
M.~Singer and W.~Stulpe.
\newblock Phase-space representations of general statistical physical theories.
\newblock {\em J. Math. Phys.}, 33:131--142, 1992.

\bibitem{ChDaJietal13}
J.~Chen, H.~Dawkins, Z.~Ji, N.~Johnston, D.~Kribs, F.~Shultz, and B.~Zeng.
\newblock Uniqueness of quantum states compatible with given measurement
  results.
\newblock {\em Phys. Rev. A}, 88:012109, 2013.

\bibitem{CaHeScTo17}
C.~Carmeli, T.~Heinosaari, J.~Schultz, and A.~Toigo.
\newblock Probing quantum state space: does one have to learn everything to
  learn something.
\newblock {\em Proc. R. Soc. A}, 473:20160866, 2017.

\bibitem{CaHeKaScTo16}
C.~Carmeli, T.~Heinosaari, A.~Karlsson, J.~Schultz, and A.~Toigo.
\newblock Verifying the {Q}uantumness of {B}ipartite {C}orrelations.
\newblock {\em Phys. Rev. Lett.}, 116:230403, 2016.

\bibitem{WoFi89}
W.K. Wootters and B.D. Fields.
\newblock Optimal state-determination by mutually unbiased measurements.
\newblock {\em Ann. Physics}, 191:363--381, 1989.

\bibitem{DuEnBeZy10}
T.~Durt, B.-G. Englert, I.~Bengtsson, and K.~Zyczkowski.
\newblock On mutually unbiased bases.
\newblock {\em Int. J. Quant. Inf.}, 8:535--640, 2010.

\bibitem{BaBoRoVa02}
S.~Bandyopadhyay, P.O. Boykin, V.~Roychowdhury, and F.~Vatan.
\newblock A new proof for the existence of mutually unbiased bases.
\newblock {\em Algorithmica}, 34:512--528, 2002.

\end{thebibliography}

\begin{thebibliography}{99}
\bibitem{MS58} Macon, N., and Spitzbart, A., Inverses of Vandermonde Matrices, {\em The American Mathematical Monthly}, Vol. 65, No. 2 (Feb., 1958), pp. 95--100.
\end{thebibliography}
\end{document}